\documentclass[acmlarge, nonacm]{acmart}

\startPage{1}

\usepackage[colorinlistoftodos,prependcaption,textsize=tiny]{todonotes}

\usepackage{longtable}
\usepackage{amsmath}
\usepackage{bbm}
\usepackage{graphicx} 
\usepackage{graphics, color}
\usepackage{algorithm}
\usepackage{algpseudocode}
\usepackage{longtable}
\usepackage{tabularx}
\usepackage{url}
\newfloat{algorithm}{t}{lop}
\usepackage{enumitem}
\usepackage{xcolor}
\usepackage{amsmath,amsthm}
\usepackage{pgf}
\usepackage{url}
\usepackage{thmtools}
\usepackage{thm-restate}
\declaretheorem{theorem}
\declaretheorem{corollary}
\declaretheorem{lemma}
\declaretheorem{proposition}
\declaretheorem{observation}
\declaretheorem{remark}
\usepackage{bbold}

\usepackage{bm}
\usepackage{mathtools}

\usepackage{xargs}                      
\algnotext{EndFor}
\algnotext{EndIf}
\algnotext{EndWhile}
\algblockdefx{Do}{EndDo}{\textbf{do}\;}{}
\algnotext{EndDo}
\algnewcommand\algorithmicswitch{\textbf{switch}}
\algnewcommand\algorithmiccase{\textbf{case}}
\algdef{SE}[SWITCH]{Switch}{EndSwitch}[1]{\algorithmicswitch\ #1\ \algorithmicdo}{\algorithmicend\ \algorithmicswitch}%
\algdef{SE}[CASE]{Case}{EndCase}[1]{\algorithmiccase\ #1}{\algorithmicend\ \algorithmiccase}%
\algtext*{EndSwitch}%
\algtext*{EndCase}%

\usepackage[colorinlistoftodos,prependcaption,textsize=tiny]{todonotes}
\newcommand{\satisfying}[1]{\ensuremath{\mathsf{Sol}(#1)}}
\newcommand{\Vars}[1]{\ensuremath{\mathsf{Vars}(#1)}}
\newcommand{\ComputeFZero}{\ensuremath{\mathsf{ComputeF0}}}
\newcommand{\PickHashFunctions}{\ensuremath{\mathsf{PickHashFunctions}}}
\newcommand{\ChooseHashFunctions}{\ensuremath{\mathsf{ChooseHashFunctions}}}

\newcommand{\EndStream}{\ensuremath{\mathsf{EndStream}}}
\newcommand{\ProcessUpdate}{\ensuremath{\mathsf{ProcessUpdate}}}
\newcommand{\ComputeEst}{\ensuremath{\mathsf{ComputeEst}}}
\newcommand{\TrailZero}{\ensuremath{\mathsf{TrailZero}}}

\newcommand{\ApproxMC}{\ensuremath{\mathsf{ApproxMC}}}
\newcommand{\ApproxMCTwo}{\ensuremath{\mathsf{ApproxMC2}}}
\newcommand{\FindMin}{\ensuremath{\mathsf{FindMin}}}
\newcommand{\AffineFindMin}{\ensuremath{\mathsf{AffineFindMin}}}
\newcommand{\FindMaxRange}{\ensuremath{\mathsf{FindMaxRange}}}
\newcommand{\thresh}{\ensuremath{\mathsf{Thresh}}}
\newcommand{\BoundedSAT}{\ensuremath{\mathsf{BoundedSAT}}}

\newcommand{\ApproxModelCountMin}{\ensuremath{\mathsf{ApproxModelCountMin}}}
\newcommand{\ApproxModelCountEst}{\ensuremath{\mathsf{ApproxModelCountEst}}}
\newcommand{\Hteop}{\ensuremath{\mathcal{H}_{\mathsf{Toeplitz}}}}
\newcommand{\Hxor}{\ensuremath{\mathcal{H}_{\mathsf{xor}}}}

\newcommand{\Bucketing}{\ensuremath{\mathsf{Bucketing}}}
\newcommand{\Estimation}{\ensuremath{\mathsf{Estimation}}}
\newcommand{\MinValue}{\ensuremath{\mathsf{Minimum}}}
\newcommand{\apsestimator}{\ensuremath{\mathsf{APS\mbox{-}Estimator}}}

\usepackage{amsmath}
\newcommand{\ignore}[1]{}


\copyrightyear{2021}
\acmYear{2021}
\setcopyright{acmcopyright}
\acmConference[PODS '21] {Proceedings of the 40th ACM SIGMOD-SIGACT-SIGAI Symposium on Principles of Database Systems}{June 20--25, 2021}{Virtual Event, China}
\acmBooktitle{Proceedings of the 40th ACM SIGMOD-SIGACT-SIGAI Symposium on Principles of Database Systems (PODS '21), June 20--25, 2021, Virtual Event, China}
\acmPrice{15.00}
\acmISBN{978-1-4503-8381-3/21/06}
\acmDOI{10.1145/3452021.3458311}

\begin{document}
\fancyhead{}

\title{Model Counting meets $F_0$ Estimation}\titlenote{The authors decided to forgo the convention of alphabetical ordering of names in favor of a randomized ordering, denoted by \textcircled{r}. The publicly verifiable record of the randomization is available at \protect\url{https://www.aeaweb.org/journals/policies/random-author-order/search} with confirmation code: XiQE7V3pKq\_A.
For citation of the work, authors request that the citation guidelines by AEA (available at \protect\url{https://www.aeaweb.org/journals/policies/random-author-order}) for random author ordering be followed.
}

\author{A. Pavan \textcircled{r}}
\email{pavan@cs.iastate.edu}
\affiliation{
	\institution{Iowa State University}
}

\author{N. V. Vinodchandran \textcircled{r}}
\email{vinod@cse.unl.edu}
\affiliation{
	\institution{University of Nebraska, Lincoln}
}

\author{Arnab Bhattacharyya \textcircled{r}}
\email{arnabb@nus.edu.sg}
\affiliation{%
	\institution{National University of Singapore}
}

\author{Kuldeep  S. Meel}
\email{meel@comp.nus.edu.sg}
\affiliation{%
	\institution{National University of Singapore}
}


\begin{abstract}

Constraint satisfaction problems (CSP's) and data stream models are two powerful abstractions to capture a wide variety of problems arising in different domains of computer science. Developments in the two communities have mostly occurred independently and with little interaction between them.
In this work, we seek to investigate whether bridging the seeming communication gap between the two communities may pave the way to richer fundamental insights. To this end, we focus on two foundational problems: model counting for CSP's and computation of zeroth frequency moments $(F_0)$ for data streams. 

Our investigations lead us to observe striking similarity in the core techniques employed in the algorithmic frameworks that have evolved separately for model counting and $F_0$ computation. We design a recipe for translation of algorithms developed for $F_0$ estimation to that of model counting, resulting in new 
algorithms for model counting. We then observe that algorithms in the context of distributed streaming can be transformed to distributed algorithms for model counting. We next turn our attention to viewing streaming from the lens of counting and show that framing $F_0$ estimation as a special case  of \#DNF counting allows us to obtain a general recipe for a rich class of streaming problems, which had been subjected to case-specific analysis in prior works. In particular, our view yields a state-of-the art algorithm for multidimensional range efficient $F_0$ estimation with a simpler analysis. 
\end{abstract}


	\maketitle

\section{Introduction}\label{sec:introduction}

{\em Constraint Satisfaction Problems} (CSP's) and the {\em data stream model} are two 
core themes in computer science with a diverse set of applications, ranging from probabilistic reasoning, networks, databases, verification, and the like. {\em Model counting} and computation of {\em zeroth frequency moment} ($F_0$) are  fundamental problems for CSP's and the data stream model respectively. This paper is motivated by our observation that despite the usage of similar algorithmic techniques for the two problems, the developments in the two communities have, surprisingly, evolved separately, and rarely has a paper from one community been cited by the other. 

Given a set of constraints $\varphi$ over a set of variables in a finite domain $\mathcal{D}$, the problem of model counting is to estimate the number of solutions of $\varphi$. We are often interested when $\varphi$ is restricted to a special class of representations such as Conjunctive Normal Form (CNF) and Disjunctive Normal Form (DNF). A data stream over a domain $[N]$ is represented by $\mathbf{a} = a_1, a_2, \cdots a_m$ wherein each item $a_i \subseteq [N]$. The {\em zeroth frequency moment}, denoted as $F_0$,  of $\mathbf{a}$ is the number of distinct elements appearing in $\mathbf{a}$, i.e.,  $|\cup_{i} a_i|$ (traditionally, $a_i$s are singletons; we will also be interested in the case when $a_i$s are sets). The fundamental nature of model counting and $F_0$ computation  over data streams has led to intense interest from theoreticians and practitioners alike in the respective communities for the past few decades. 

The starting point of this work is the confluence of two viewpoints. The first viewpoint contends that some of the algorithms for model counting can conceptually be thought of as operating on the stream of the solutions of the constraints. The second viewpoint contends that a stream can be viewed as a DNF formula, and the problem of $F_0$ estimation is similar to model counting. These viewpoints make it natural to believe that algorithms developed in the streaming setting can be directly applied to model counting, and vice versa. We explore this connection and indeed, design new algorithms for model counting inspired by algorithms for estimating $F_0$ in data streams. By exploring this connection further, we design new  algorithms to estimate $F_0$ for streaming sets that are succinctly represented by constraints. To put our contributions in context, we briefly survey the historical development of algorithmic frameworks in both model counting and $F_0$ estimation and point out the similarities.

\subsection*{Model Counting}
The complexity-theoretic study of model counting was initiated by Valiant who showed that this  problem, in general, is \#P-complete ~\cite{Valiant79}.   This motivated researchers to investigate approximate model counting and in particular achieving $(\varepsilon,\delta)$-approximation schemes.  The complexity of approximate model counting depends on its representation. When the model $\varphi$  is represented as a CNF formula $\varphi$, designing an efficient $(\varepsilon,\delta)$-approximation is  NP-hard \cite{Stockmeyer83}. In contrast, when it is represented as a DNF formula, model counting admits FPRAS (fully polynomial-time approximation scheme)~\cite{KL83,KLM89}. We will use \#CNF to refer to the case when $\varphi$ is a CNF formula while \#DNF to refer to the case when $\varphi$ is a DNF formula. 

For \#CNF, Stockmeyer~\cite{Stockmeyer83} provided a hashing-based randomized procedure that can compute ($\varepsilon,\delta)$-approximation within time polynomial in $|\varphi|, \varepsilon,\delta$,  given access to an NP oracle.  Building on Stockmeyer's approach and motivated by the unprecedented breakthroughs in the design of SAT solvers, researchers have proposed a series of algorithmic improvements that have allowed the hashing-based techniques for approximate model counting to scale to formulas involving hundreds of thousands of variables~\cite{GSS07,CMV13b,EGSS13c,IMMV15,CDM15,CMV16,AT17,SM19,SGM20}. The practical implementations substitute NP oracle with SAT solvers. In the context of model counting, we are primarily interested in time complexity and therefore, the number of NP queries is of key importance. The emphasis on the number of NP  calls also stems from practice as the practical implementation of model counting algorithms have shown to spend over 99\% of their time in the underlying SAT calls~\cite{SM19}.

Karp and Luby~\cite{KL83} proposed the first FPRAS scheme for \#DNF, which was subsequently improved in the follow-up works~\cite{KLM89,DKLR00}.  Chakraborty, Meel, and Vardi~\cite{CMV16} demonstrated  that the hashing-based framework can be extended to \#DNF, hereby providing a unified framework for both \#CNF and \#DNF. Meel, Shrotri, and Vardi~\cite{MSV17,MSV18,MSV19} subsequently improved the complexity of the hashing-based approach for \#DNF and observed that hashing-based techniques achieve better scalability than that of Monte Carlo techniques. 

\subsection*{Zeroth Frequency Moment Estimation}
Estimating $(\varepsilon,\delta)$-approximation of the $k^{\rm th}$ frequency moments ($F_k$) is a central problem in the data streaming model~\cite{AMS99}. In particular,  considerable work has been done in designing algorithms for estimating the $0^{th}$ frequency moment ($F_0$), the number of distinct elements in the stream. While designing streaming algorithms, the primary resource concerns are two-fold: space complexity and processing time per element. For an algorithm to be considered efficient,  these should be ${\rm poly}(\log N,1/\epsilon)$ where $N$ is the size of the universe~\footnote{We ignore $O(\log {1 \over \delta})$ factor in this discussion}. 

The first algorithm for computing $F_0$ with a constant factor approximation was proposed by Flajolet and Martin, who assumed the existence of hash functions with ideal properties resulting in an algorithm with undesirable space complexity~\cite{FM85}. In their seminal work, Alon, Matias, and Szegedy designed  an $O(\log N)$ space algorithm for $F_0$ with a constant approximation ratio that employs 2-universal hash functions~\cite{AMS99}. Subsequent investigations into hashing-based schemes by Gibbons and Tirthapura~\cite{GT01} and Bar-Yossef, Kumar, and Sivakumar~\cite{BKS02} provided $(\varepsilon, \delta)$-approximation algorithms with space and time complexity $\log N \cdot {\rm poly} ({1\over \varepsilon})$. 
Subsequently, Bar-Yossef et al. proposed {\em three algorithms} with improved space and time complexity~\cite{BJKST02}. While the three algorithms employ hash functions, they differ conceptually in the usage of relevant random variables for the estimation of $F_0$. 
This line of work resulted in the development of an algorithm with optimal space complexity ${O}(\log N + {1\over \varepsilon^2})$ and $O(\log N)$ update time \cite{KNW10}. 

The above-mentioned works are in the setting where each data item $a_i$ is an element of the universe.
Subsequently, there has been a series of results of estimating $F_0$ in  rich scenarios with particular focus to handle the cases $a_i \subseteq [N]$ such as a list or a multidimensional range~\cite{BKS02, PT07, TW12,SP09}.

\ignore{
The first algorithm is based on the idea that if we hash all the items of the stream, then $\mathcal{O}(1/\varepsilon^2)$-th minimum of the hash valued is a good estimator of $F_0$. The second algorithm  chooses a set of  k functions, $\{h_1, h_2, \ldots \}$,  such that each $h_j$ is picked randomly from $\mathcal{O}(1/\varepsilon^2)$-independent hash family. For each hash function $h_j$, we say that $h_j$ is not {\em lonely} if there exists $a_i \in \mathbf{a}$ such that $h_j(a_i) = 0$. One can then estimate $F_0$ of $\mathbf{a}$ by estimating the number of hash functions that are not lonely. Finally, the third algorithm ... \todo{@Vinod: Can you add description of bucketing algorithm using the sub-sampling terms...}. We will refer to the above three algorithms based on their strategies: Minimum, Estimation, and Bucketing. 
}

\subsection*{The Road to a Unifying Framework}

As mentioned above, the algorithmic developments for model counting and $F_0$ estimation 
have largely relied on the usage of hashing-based techniques and yet these developments have, surprisingly, been separate, and rarely has a work from one community been cited by the other. In this context, we wonder whether it is possible to bridge this gap and if such an exercise would contribute to new algorithms for model counting as well as for $F_0$ estimation? The main conceptual contribution of this work is an affirmative answer to the above question. First, we point out that the two well-known algorithms; Stockmeyer's \#CNF algorithm ~\cite{Stockmeyer83} that is further refined  by 
Chakraborty et. al. ~\cite{CMV16} and Gibbons and Tirthapura's  $F_0$ estimation algorithm~\cite{GT01}, are essentially the same. 

The core idea of the hashing-based technique of Stockmeyer's and Chakraborty et al's scheme is to use pairwise independent hash functions to partition the solution space (satisfying assignments of a CNF formula) into {\em roughly equal and small } cells, wherein a cell is {\em small} if the number of solutions is less than a pre-computed threshold, denoted by {\thresh}. Then a good estimate  for the number of solutions is the {\em number of solutions in an arbitrary cell} $\times$ {\em number of cells}.  To partition the solution space, pairwise independent hash functions are used.
To determine the appropriate number of cells, the solution space is iteratively partitioned  as follows. At the $m^{th}$ iteration, a hash function with range $\{0,1\}^m$ is considered resulting in cells $h^{-1}(y)$ for each $y\in \{0,1\}^m$.  An NP oracle can be employed to check whether a particular cell (for example $h^{-1}(0^m)$) is small by enumerating solutions one by one until we have either obtained $\thresh$+1 number of solutions or we have exhaustively enumerated all the solutions. If the the cell  $h^{-1}(0^m)$ is small, then the algorithm outputs $t\times 2^m$ as an estimate where $t$ is the number of solutions in the cell $h^{-1}(0^m)$. If the cell $h^{-1}(0^m)$ is not small,  then the algorithm moves on to the next iteration where a hash function with  range $\{0,1\}^{m+1}$ is considered.

\ignore{
 wherein we first invoke an NP oracle to determine if the number of solutions of the formula is less than {\thresh}. We then partition the space of solution into two cells and choose a cell randomly by conjuncting the formula with a random XOR equation. We continue to add random XORs until the chosen cell is indeed small, and at this point, the  the estimate of the number of solutions is simply the product of the number of cells ($2^m$), where $m$ is the number of XORs added, and the number of solutions in the chosen cell. The number of invocations of NP oracles is often the primary concern from algorithmic point of view. 
}

We now describe Gibbons and Tirthapura's algorithm for $F_0$ estimation which we call the {\Bucketing} algorithm. We will assume the universe $[N] = \{0,1\}^n$. The algorithm maintains a bucket of size $\thresh$ and starts by picking a hash function  $h:\{0,1\}^n \rightarrow \{0,1\}^n$. It iterates over  sampling levels. At level $m$, when a data item $x$ comes, if $h(x)$ starts with $0^m$,  then $x$ is added to the bucket. If the bucket overflows, then the sampling level is increased to $m+1$ and all  elements $x$ in the bucket other than the ones with $h(x)=0^{m+1}$ are deleted. At the end of the stream, the value $t\times 2^{m}$ is output as the estimate where $t$ is the number of elements in the bucket and $m$ is the sampling level. 

These two algorithms are conceptually the same. In the \Bucketing\ algorithm, at the sampling level $m$, it looks at only the first $m$ bits of the hashed value; this is equivalent to considering a hash function with range $\{0,1\}^m$. Thus the bucket is nothing but all the elements in the stream that belong to the cell $h^{-1}(0^m)$.  
The final estimate is  the number of  elements in the bucket times the   number of cells, identical to Chakraborty et. al's algorithm. In both algorithms, to obtain an $(\varepsilon, \delta)$ approximation, the \thresh~value is chosen as $O({1 \over \varepsilon^2})$ and the median of $O(\log {1\over \delta})$ independent estimations is output. 
\subsection*{Our Contributions}

Motivated by the conceptual identity between the two algorithms, we further explore the connections between algorithms for model counting and $F_0$ estimation. 

\begin{enumerate}
	\item We  formalize a recipe to transform streaming algorithms for $F_0$ estimation to those for model counting. Such a transformation yields new $(\varepsilon, \delta)$-approximate algorithms for model counting, which are different from currently known algorithms. Recent studies in the fields of automated reasoning have highlighted the need for diverse approaches~\cite{XHHL08}, and similar studies in the context of \#DNF provided strong evidence to the power of diversity of approaches~\cite{MSV18}. In this context, these newly obtained algorithms open up several new interesting directions of research ranging from the development of MaxSAT solvers with native XOR support to open problems in designing FPRAS schemes. 
	
	\item Given the central importance of \#DNF (and its weighted variant) due to a recent surge of interest in scalable techniques for provenance in probabilistic databases~\cite{S18,S19}, a natural question is whether one can design efficient techniques in the distributed setting. In this work, we initiate the study of distributed \#DNF. We then show that the  transformation recipe from $F_0$ estimation to model counting allows us to view the problem of the design of distributed \#DNF algorithms through the lens of {\em distributed functional monitoring} that is well studied in the data streaming literature. 
	
	\item  Building upon the connection between model counting and $F_0$ estimation, we design new algorithms to estimate $F_0$ over {\em structured set streams} where each element of the stream is a  (succinct representation of a) subset of the universe. Thus, the stream is $S_1, S_2, \cdots $ where each $S_i \subseteq [N]$ and the goal is to estimate the $F_0$ of the stream, i.e.~size of  $\cup_{i} S_i$. In this scenario,  a traditional $F_0$ streaming algorithm that processes each element of the set incurs high per-item processing time-complexity and is inefficient. Thus the goal is to design algorithms whose per-item time (time to process each $S_i$) is poly-logarithmic in the size of the universe. Structured set streams that are considered in the literature include  1-dimensional and multidimensional ranges~\cite{PT07,TW12}. Several interesting problems such as max-dominance norm~\cite{CM03}, counting triangles in graphs~\cite{BKS02}, and distinct summation problem~\cite{CLKB04} can be reduced to computing $F_0$ over such ranges. 
	
\smallskip	
We observe that several structured sets can be represented as small DNF formulae and thus  $F_0$ counting  over these structured set data streams can be viewed as a special case of \#DNF.  Using the hashing-based techniques for \#DNF, we  obtain a general recipe for a rich class of  structured sets that include multidimensional ranges, multidimensional arithmetic progressions, and affine spaces.  Prior work on single and multidimensional ranges~\footnote{Please refer to Remark \ref{remark} in Section \ref{sec:multidimrange} for a discussion on the earlier work on multidimensional ranges~\cite{TW12}.} had to rely on involved analysis for each of the specific instances, while our work provides a general recipe for both analysis and implementation.  
	
\end{enumerate}

\subsection*{Organization}

We present notations and preliminaries in Section~\ref{sec:prelims}. We then present the transformation of $F_0$ estimation to model counting in Section~\ref{sec:counting}. We then focus on distributed \#DNF in Section~\ref{sec:distributed}. We then present the transformation of model counting algorithms to structured set streaming algorithms in Section~\ref{sec:streaming}. We conclude in  Section~\ref{sec:conclusion} with a discussion of future research directions. 

We would like to emphasize that the primary objective of this work is to provide a unifying framework for $F_0$ estimation and model counting. Therefore, when designing new algorithms based on the transformation recipes, we intentionally focus on conceptually cleaner algorithms and leave potential improvements in time and space complexity for future work. 

\section{Notation}\label{sec:prelims}

We will use assume the universe $[N] = \{0,1\}^n$.
We write $\Pr\left[\mathcal{Z}: {\Omega} \right]$ to denote the probability of
outcome $\mathcal{Z}$ when sampling from a probability space ${\Omega}$.  For
brevity, we omit ${\Omega}$ when it is clear from the context. 

\subsubsection*{$F_0$ Estimation}
A data stream $\mathbf{a}$ over domain $[N]$ can be represented as $\mathbf{a} = a_1, a_2, \ldots a_m$ wherein each item $a_i \in [N]$. Let $\mathbf{a}_u = \cup_{i} \{a_i\}$. 
$F_0$ of the stream $\mathbf{a}$ is $|\mathbf{a}_u|$. We are often interested in a a {\em probably approximately correct} scheme that returns an {\em $(\varepsilon,\delta)$-estimate} $c$, i.e., 
\begin{align*}
\Pr\left[\frac{|\mathbf{a}_u|}{1+\varepsilon} \leq c \leq (1+\varepsilon) |\mathbf{a}_u| \right] \geq 1-\delta
\end{align*}

\subsubsection*{Model Counting}

Let $\{x_1, x_2, \ldots x_n\}$ be a set of Boolean variables.
For a Boolean formula  $\varphi$, let $\Vars{\varphi}$ denote the set of variables appearing in $\varphi$.  Throughout the paper, unless otherwise stated, we will assume that the relationship $n = |\Vars{\varphi}|$ holds. 
We denote the set of all satisfying assignments  of $\varphi$ by $\satisfying{\varphi}$. 

The \emph{propositional model counting problem} is to compute
$|\satisfying{\varphi}|$ for a given formula $\varphi$.  A \emph{probably approximately correct} (or PAC) counter is a probabilistic algorithm ${\mathsf{ApproxCount}}(\cdot, \cdot,\cdot)$ that takes as inputs a formula $\varphi$,  a tolerance $\varepsilon>0$, and a confidence  $\delta\in (0, 1]$, and returns a $(\varepsilon,\delta)$-estimate $c$, i.e., 
\begin{align*}
 \Pr\Big[\frac{|\satisfying{\varphi}|}{1+\varepsilon} \le c \le (1+\varepsilon)|\satisfying{\varphi}|\Big] \ge 1-\delta.
\end{align*}
 PAC guarantees are also sometimes referred to as $(\varepsilon,\delta)$-guarantees. We use \#CNF (resp. \#DNF) to refer to the model counting problem when $\varphi$ is represented as CNF (resp. DNF).

\subsubsection*{k-wise Independent hash functions}

Let $n,m\in \mathbb{N}$ and $\mathcal{H}(n,m) \triangleq \{ h:\{0,1\}^{n} \rightarrow \{0,1\}^m \}$ be a family of hash functions mapping $\{0,1\}^n$ to $\{0,1\}^m$. We use $h \xleftarrow{R} \mathcal{H}(n,m)$ to denote the probability space obtained by choosing a function $h$ uniformly at random from $\mathcal{H}(n,m)$. 
\begin{definition}
	A family of hash functions $\mathcal{H}(n,m)$ is $k-$wise independent if $\forall \alpha_1, \alpha_2, \ldots \alpha_k \in \{0,1\}^m$,  $\text{ distinct } x_1, x_2, \ldots x_k \in \{0,1\}^n, h \xleftarrow{R} \mathcal{H}(n,m)$, 
	\begin{align}
	\Pr[(h(x_1) = \alpha_1) \wedge (h(x_2) = \alpha_2) \ldots (h(x_k) = \alpha_k) ] = \frac{1}{2^{km}} %
	\end{align}
\end{definition}	

We will use $\mathcal{H}_{\mathsf{k-wise}}(n,m)$ to refer to a $k-$wise independent family of hash functions mapping $\{0,1\}^n$ to $\{0,1\}^m$.

\subsubsection{Explicit families}	
In this work, one hash family of particular interest is $\Hteop(n,m)$, which is known to be  2-wise independent~\cite{carter1977universal}. The family is defined as follows:
$\Hteop(n,m) \triangleq \{ h: \{0,1\}^n \rightarrow \{0,1\}^m  \}$ is the family of functions of the form $h(x) = Ax+b$ with $A \in \mathbb{F}_{2}^{m \times n}$ and $b \in \mathbb{F}_{2}^{m \times 1}$ where $A$ is a uniformly randomly chosen Toeplitz matrix of size $m\times n$ while $b$ is uniformly randomly  matrix of size $m \times 1$. 
Another related hash family of interest is $\Hxor(n,m)$ wherein $h(X)$ is again of the form $Ax+b$ where $A$ and $b$ are uniformly randomly chosen matrices of sizes $m \times n$ and $m \times 1$ respectively.  Both $\Hteop$ and $\Hxor$ are 2-wise independent but it is worth noticing that $\Hteop$ can be represented with $\Theta(n)$-bits while $\Hxor$ requires $\Theta(n^2)$ bits of representation.

	For every $m \in \{1, \ldots
	n\}$, the $m^{th}$ prefix-slice of $h$, denoted $h_{m}$, is a
	map from $\{0,1\}^{n}$ to $\{0,1\}^m$, where $h_{m}(y)$ is the first $m$ bits of $h(y)$. Observe that when $h(x) = Ax+b$,  $h_{m}(x) = A_{m}x+b_{m}$, where $A_{m}$ denotes the submatrix formed by the first $m$ rows of $A$ and $b_{m}$ is the first $m$ entries of the vector $b$.


\section{From Streaming to Counting}\label{sec:counting}


As a first step, we present a unified view of the three hashing-based algorithms proposed in 
Bar-Yossef et al \cite{BJKST02}.  The first algorithm is the \Bucketing\ algorithm discussed above with the observation that instead of keeping the elements in the bucket, it suffices to keep their hashed values. Since in the context of model counting, our primary concern is with time complexity, we will focus on Gibbons and Tirthapura's \Bucketing\ algorithm in \cite{GT01} rather than Bar-Yossef et al.'s modification. 
The second algorithm, which we call \MinValue, is based on the idea that if we hash all the items of the stream, then $\mathcal{O}(1/\varepsilon^2)$-th minimum of the hash valued can be used compute a good estimate of $F_0$. The third algorithm, which we call \Estimation,  chooses a set of  $k$ functions, $\{h_1, h_2, \ldots \}$,  such that each $h_j$ is picked randomly from an $\mathcal{O}(\log(1/\varepsilon))$-independent hash family. For each hash function $h_j$, we say that $h_j$ is not {\em lonely} if there exists $a_i \in \mathbf{a}$ such that $h_j(a_i) = 0$. One can then estimate $F_0$ of $\mathbf{a}$ by estimating the number of hash functions that are not lonely. 

Algorithm~\ref{alg:hashstream}, called $\ComputeFZero$,  presents the overarching 
architecture of the three proposed algorithms. Each of these algorithms first picks 
an appropriate set of hash functions $H$  and initializes the sketch $\mathcal{S}$.   The architecture of {\ComputeFZero} is fairly simple:  it chooses a collection of hash functions using 
 \ChooseHashFunctions, calls the subroutine {\ProcessUpdate} for every incoming element of the stream, and invokes \ComputeEst\ at the end 
 of the stream to return the $F_0$ approximation.

 \subsubsection*{\ChooseHashFunctions} As shown in Algorithm \ref{alg:choosehash}, the hash functions depend on the strategy being implemented. The subroutine $\PickHashFunctions(\mathcal{H}, t)$ returns 
 a collection of $t$ independent hash functions from the family $\mathcal{H}$.  We use $H$ to denote the collection of hash functions  returned, this collection viewed as either 1-dimensional array or as a 2-dimensional array. When $H$ is 1-dimensional array, $H[i]$ to denote the $i$th hash function of the collection and when $H$ is a 2-dimensional array $H[i][]j$ is the $[i, j]$th hash functions.

\subsubsection*{Sketch Properties}

For each of the three algorithms, their corresponding sketches can be viewed as arrays of size of $35\log (1/\delta)$. The parameter $\thresh$ is set to $96/\varepsilon^2$.

\begin{description}
	\item[\Bucketing] The element $\mathcal{S}[i]$ 	is a tuple  $\langle \ell_i, m_i\rangle$ where $\ell_i$ is a list of size at most $\thresh$, where $\ell_i = \{x  \in \mathbf{a}\mid H[i]_{m_i}(x)= 0^{m_i}\}$.  We use $\mathcal{S}[i](0)$ to denote $\ell_i$ and $\mathcal{S}[i](1)$ to denote $m_i$.
	
	\item[\MinValue] The element $\mathcal{S}[i]$ holds the lexicographically distinct $\thresh$  many smallest elements of $\{H[i](x)~|~x \in \mathbf{a}\}$. 
	
	\item[\Estimation] The element $\mathcal{S}[i]$ holds a tuple of size $\thresh$. The $j$'th entry of this tuple is the largest number of trailing zeros in any element of $H[i,j](\mathbf{a})$.  	
\end{description}

 \subsubsection*{\ProcessUpdate} 
 For a new item $x$, the update of $\mathcal{S}$, as shown in Algorithm \ref{alg:processupdate} is as follows: 
\begin{description}
	\item[Bucketing]  For a new item 
		$x$, if $H[i]_{m_i}(x) = 0^{m_i}$, then we add it to $\mathcal{S}[i]$ if $x$ is not already present in $\mathcal{S}[i]$. If the size of $\mathcal{S}[i]$ is greater than $\thresh$ (which is set to be $\mathcal{O}(1/\varepsilon^2)$), then
we increment the $m_i$ as  in line~\ref{line:update-increase-prefix}.

\item[Minimum] For a new item $x$, if $H[i](x)$ is smaller than the $\max{\mathcal{S}[i]}$, then we replace $\max{\mathcal{S}[i]}$ 
with $H[i](x)$. 

\item[Estimation] For a new item $x$, compute $z = {\TrailZero(H[i,j](x))}$, i.e, the number of trailing zeros in $H[i,j](x)$,  and replace $\mathcal{S}[i,j]$ with $z$ if $z$ is larger than $\mathcal{S}[i,j]$.
\end{description}

\subsubsection*{\ComputeEst}
 Finally, for each of the algorithms, we estimate $F_0$ based on the sketch $\mathcal{S}$
as described in the subroutine {\ComputeEst} presented as Algorithm~\ref{alg:computest}. 
It is crucial to note that the estimation of $F_0$ is performed solely using the sketch $\mathcal{S}$ for
the Bucketing and Minimum algorithms. The Estimation algorithm requires an additional parameter $r$ that depends on a loose estimate of $F_0$;
we defer details to Section \ref{sec:estim}.

\begin{algorithm}[htb]
	\caption{$\ComputeFZero(n,\varepsilon, \delta)$}
	\label{alg:hashstream}
	\begin{algorithmic}[1]
		\State $\thresh \gets {96}/{\varepsilon^2}$
		\State $t \gets 35 \log(1/\delta)$
		\State $H \gets \ChooseHashFunctions(n,\thresh,t)$
		\State $\mathcal{S} \gets \{\}$
		\While{true}
		\If{$\EndStream$} exit;
		\EndIf
		\State $x \gets input()$
		\State $\ProcessUpdate(\mathcal{S}, H, x, \thresh)$
		\EndWhile
		\State $Est \gets \ComputeEst(\mathcal{S},\thresh)$
		\State \Return $Est$
	\end{algorithmic}
\end{algorithm}

\begin{algorithm}
	\caption{\ChooseHashFunctions($n, \thresh, t$)}
	\label{alg:choosehash}
	\begin{algorithmic}[1]
	    \Switch{AlgorithmType}
	    \Case{AlgorithmType==\Bucketing}
	    \State $H \gets \PickHashFunctions(\Hteop(n,n),t)$
	    \EndCase
		\Case{AlgorithmType==\MinValue}
		\State $H \gets \PickHashFunctions(\Hteop(n,3n),t)$
		\EndCase	
		\Case{AlgorithmType==\Estimation}
			   \State $s \gets 10 \log(1/\varepsilon)$
		\State $H \gets \PickHashFunctions(\mathcal{H}_{s-{\rm wise}}(n,n), t \times \thresh)$
		\EndCase
	    \EndSwitch
	    \Return $H$
	\end{algorithmic}
	
\end{algorithm}

\begin{algorithm}
	\caption{$\ProcessUpdate(\mathcal{S}, H, x, \thresh)$}
	\label{alg:processupdate}
	\begin{algorithmic}[1]
		\For {$i \in [1, |H|$]}
		\Switch{AlgorithmType}
			\Case{Bucketing}
			\State $m_i  =  \mathcal{S}[i](0)$
	\If {$H[i]_{m_i}(x) == 0^{m_i}$}
	\State $\mathcal{S}[i](0) \gets \mathcal{S}[i](0) \cup \{x\}$

			\If {$\mathrm{size}(\mathcal{S}[i](0)) > \thresh$}
				\State $\mathcal{S}[i](1) \gets \mathcal{S}[i](1) + 1$ \label{line:update-increase-prefix}
				\For{$y  \in \mathcal{S}$}
					\If {$H[i]_{m_i+1}(y) \neq 0^{m_i+1}$}
					\State Remove($\mathcal{S}[i](0), y$)
					\EndIf	
				\EndFor
			\EndIf
	\EndIf
	\EndCase
		\Case{Minimum}
\If {$\mathrm{size}(\mathcal{S}[i]) < \thresh$}
\State $\mathcal{S}[i].\mathrm{Append}(H[i](x))$
\Else
\State $j \gets \arg\max(S[i])$

\If {$\mathcal{S}[i](j) > H[i](x)$}
\State $\mathcal{S}[i](j)  \gets  H[i](x)$
\EndIf
\EndIf
	\EndCase
		\Case{Estimation}
		\For{$j \in [1, \thresh]$}
		\State $S[i,j] \gets \max(S[i,j], \TrailZero(H[i,j](x)))$	
		\EndFor	
		\EndCase		

		\EndSwitch	
		\EndFor 
		\State \Return $\mathcal{S}$
	\end{algorithmic}
\end{algorithm}

\begin{algorithm}
	\caption{\ComputeEst($\mathcal{S}, \thresh$)}
	\label{alg:computest}
	\begin{algorithmic}[1]
	
		\Switch{AlgorithmType}
	\Case{Bucketing}
	\State \Return  $\mathrm{Median}\left(\left\{ \mathrm{size}(\mathcal{S}[i](0))\times 2^{\mathcal{S}[i](1)}\right\}_{i}\right)$
	\EndCase
	\Case{Minimum}
	\State \Return $\mathrm{Median}\left(\left\{ \frac{\thresh \times 2^m }{\max \{\mathcal{S}[i] \} } \right\}_{i}   \right)$
\EndCase
	\Case{Estimation($r$)}
	\State \Return $\mathrm{Median}\left(\left\{\frac{\ln\left(1-\frac{1}{\thresh}\sum_{j=1}^{\thresh} \mathbb{1}\{\mathcal{S}[i,j]\geq r\}\right)}{\ln(1-2^{-r})}\right\}_i\right)$
\EndCase
\EndSwitch
	\end{algorithmic}
	
\end{algorithm}


%

\subsection{A Recipe For Transformation}
Observe that for each of the algorithms, the final computation of $F_0$ estimation depends on the sketch $\mathcal{S}$. Therefore, as long as for two streams $\mathbf{a}$ and $\hat{\mathbf{a}}$, if their corresponding sketches, say $\mathcal{S}$ and $\hat{\mathcal{S}}$ respectively, are equivalent, the three schemes presented above would return the same estimates. 
The recipe for a transformation of streaming algorithms to model counting algorithms is based on the following insight: 

	\begin{enumerate}
		\item Capture the relationship $\mathcal{P} (\mathcal{S}, H, \mathbf{a}_{u})$ between the sketch $\mathcal{S}$, set of hash functions $H$, and  set $\mathbf{a}_{u}$ at the end of stream.  Recall that $\mathbf{a}_{u}$ is the set of all distinct elements of the stream $\mathbf{a}$.
		\item The formula $\varphi$ is viewed as symbolic representation of the unique set $\mathbf{a}_{u}$ represented by the stream $\mathbf{a}$ such that $\satisfying{\varphi} = \mathbf{a}_{u}$.
		\item Given a formula $\varphi$ and set of hash functions $H$, design an algorithm to construct sketch $\mathcal{S}$ such that $\mathcal{P} (\mathcal{S}, H, \satisfying{\varphi})$ holds. And now, we can estimate $|\satisfying{\varphi}|$ from $\mathcal{S}$. 
	\end{enumerate}
	
In the rest of this section, we will apply the above recipe to the three types of $F_0$ estimation algorithms, and derive corresponding model counting algorithms. In particular, we show how applying the above recipe to the \Bucketing\ algorithm leads us to reproduce the state of the art hashing-based model counting algorithm, {\ApproxMC}, proposed by Chakraborty et al \cite{CMV16}. Applying the above recipe to \MinValue\ and \Estimation\ allows us to obtain fundamentally different schemes. In particular, we observe while model counting algorithms based on \Bucketing\ and \MinValue\ provide FPRAS's when $\varphi$ is DNF, such is not the case for the algorithm derived based on \Estimation.


\subsection{Bucketing-based Algorithm}\label{sec:bucketing}

The \Bucketing\ algorithm chooses a set $H$ of pairwise independent hash functions and  maintains a sketch $\mathcal{S}$ that we will describe.  Here we use $\Hteop$ as our choice of pairwise independent hash functions. 
The sketch $\mathcal{S}$ is an array where, each $\mathcal{S}[i]$ of the form $\langle c_i, m_i\rangle$. We say that the relation $\mathcal{P}_1 (\mathcal{S}, H, \mathbf{a}_u)$ holds if 
\begin{enumerate}
	\item $|\mathbf{a}_{u} \cap \{x~|~H[i]_{m_i-1}(x) = 0^{m_i-1}\}| \geq \frac{96}{\varepsilon^2}$ 
	\item $c_i = |\mathbf{a}_{u} \cap \{x~|~H[i]_{m_i}(x) = 0^{m_i}\}| < \frac{96}{\varepsilon^2}$
	
\end{enumerate}

The following lemma due to Bar-Yossef {\it et al.}~\cite{BJKST02} and Gibbons and Tirthapura~\cite{GT01} captures the relationship among the sketch $\mathcal{S}$, the relation $\mathcal{P}_1$ and the number of distinct elements of a multiset. 

\begin{lemma}~\cite{BJKST02,GT01}\label{lem:bucket}
Let $\mathbf{a} \subseteq \{0,1\}^n$ be a multiset and $H \subseteq \Hteop(n,n)$ where and each $H[i]$s  are independently drawn and $|H|=O(\log 1/\delta)$ and let $\mathcal{S}$ be such that the $\mathcal{P}_1 (\mathcal{S}, H, a_u)$ holds. Let $c =  \mbox{\rm Median }\{c_i \times 2^{m_i}\}_i$.
Then

\[\Pr \left[ \frac{|\mathbf{a}_u|}{(1+\varepsilon)} \leq c
\leq (1+\varepsilon)|\mathbf{a}_u|\right]\geq 1- \delta.\]
\end{lemma}

	

To design an algorithm for model counting, based on the bucketing strategy, we turn to the subroutine introduced by Chakraborty, Meel, and Vardi: {\BoundedSAT}, whose properties are formalized as follows:

\begin{proposition}\label{prop:boundedsat}\cite{CMV13b,CMV16}
 There is an algorithm {\BoundedSAT} that gets  $\varphi$ over $n$ variables,  a hash function $h \in \Hteop (n, m)$, and a number $p$ as inputs, returns $\min (p, |\satisfying{\varphi \wedge h(x) = {0^m}}|)$.  If $\varphi$ is a CNF formula, then {\BoundedSAT} makes $\mathcal{O}(p)$ calls to a NP oracle. If $\varphi$ is a DNF formula with $k$ terms, then {\BoundedSAT} takes $\mathcal{O}( n^3  \cdot k  \cdot p)$ time.  
\end{proposition}

	
Equipped with Proposition~\ref{prop:boundedsat}, we now turn to designing an  algorithm for model counting based on the Bucketing strategy. The algorithm follows in similar fashion to its streaming counterpart where $m_i$  is iteratively incremented until the number of solutions of the formula ($\varphi \wedge H[i]_{m_i}(x) = 0^{m_i})$ is less than $\thresh$. Interestingly, an approximate model counting algorithm, called {\ApproxMC}, based on bucketing strategy was discovered independently by Chakraborty et. al.~\cite{CMV13b} in 2013. We reproduce an adaptation {\ApproxMC} in Algorithm~\ref{alg:approxmc} to showcase how {\ApproxMC} can be viewed as transformation of the \Bucketing\ algorithm. In the spirit of \Bucketing, {\ApproxMC} seeks to construct a sketch $\mathcal{S}$ of size $t \in \mathcal{O}(\log (1/\delta))$. To this end, for every iteration of the loop, we continue to increment the value of the loop until the conditions specified by the relation $\mathcal{P}_1 (\mathcal{S}, H, \satisfying{\varphi})$ are met. For every iteration $i$, the estimate of the model count is  $c_i \times 2^{m_i}$. Finally, the estimate of the model count is simply the median of the estimation of all the iterations. Since in the context of model counting, we are concerned with time complexity, wherein both $\Hteop$ and $\Hxor$ lead to same time complexity. Furthermore, Chakraborty et al.~\cite{CMV13a} observed no difference in empirical runtime behavior due to $\Hteop$ and $\Hxor$.

 The following theorem establishes the correctness of {\ApproxMC}, and the proof follows from Lemma~\ref{lem:bucket} and Proposition~\ref{prop:boundedsat}.  

\begin{theorem}
Given a formula $\varphi$, $\varepsilon$, and $\delta$, {\ApproxMC} returns
$Est$ such that $\Pr [ \frac{|\satisfying{\varphi}|}{1+\varepsilon} \leq Est \leq (1+\varepsilon)|\satisfying{\varphi}|] \geq 1- \delta$.
If $\varphi$ is a CNF formula, then this algorithm makes $\mathcal{O}(n \cdot \frac{1}{\varepsilon^2} \log (1/\delta))$ calls to NP oracle. If $\varphi$ is a DNF formula then {\ApproxMC} is FPRAS. In particular for a DNF formula with $k$ terms, {\ApproxMC} takes $\mathcal{O}( n^4  \cdot k  \cdot  \frac{1}{\varepsilon^2} \cdot \log (1/\delta))$ time. 
\end{theorem} 

\begin{algorithm}
	\caption{$\ApproxMC (\varphi,\varepsilon, \delta)$}
	\label{alg:approxmc}
	\begin{algorithmic}[1]
		\State $t \gets  35 \log (\frac{1}{\delta})$
		\State $H \gets \PickHashFunctions(\Hteop(n,n),t)$
		\State $\mathcal{S} \gets \{\}$;  
		\State $\thresh \gets \frac{96}{\varepsilon^2}$
		\For {$i \in  [1, t] $}\label{line:iteration-loop-begin}
			\State $m_i \gets 0$
			\State $c_i \gets \BoundedSAT(\varphi,  H[i]|_{m_i}, \thresh)$
			\While{$c_i \geq \thresh$ } \label{line:search-loop-begin}
				\State $m_i \gets m_i+1$
				\State $c_i \gets \BoundedSAT(\varphi, H[i]|_{m_i}(x), \thresh)$
			
			\EndWhile  \label{line:search-loop-end}
		\State $\mathcal{S}[i] \gets (c_i, m_i)$
		\EndFor \label{line:iteration-loop-end}
		\State $Est \gets Median(\{ \mathcal{S}[i](0)\times 2^{\mathcal{S}[i](1)}\}_{i})$
		\State \Return $Est$
	\end{algorithmic}
\end{algorithm}

\subsubsection*{Further Optimizations}
We now discuss how the setting of model counting allows for further optimizations. 
Observe that for all $i$, $\satisfying{\varphi \wedge (H[i]_{m_i-1})(x) = 0^{m_i-1}} \subseteq \satisfying{\varphi \wedge (H[i]_{m_i})(x) = 0^{m_i}}$. Note that we are interested in finding the value of $m_i$ such that $|\satisfying{\varphi \wedge (H[i]_{m_i-1})(x) = 0^{m_i-1}} | \geq \frac{96}{\varepsilon^2}$ and $|\satisfying{\varphi \wedge (H[i]_{m_i})(x) = 0^{m_i}} | < \frac{96}{\varepsilon^2}$, therefore, we can perform a binary search for $m_i$ instead of a linear search performed in the loop~\ref{line:search-loop-begin}--~\ref{line:search-loop-end}. Indeed, this observation was at the core of Chakraborty et al's followup work~\cite{CMV16}, which proposed the ApproxMC2, thereby reducing the number of calls to NP oracle from  $\mathcal{O}(n \cdot \frac{1}{\varepsilon^2} \log (1/\delta))$ to  $\mathcal{O}(\log n \cdot \frac{1}{\varepsilon^2} \log (1/\delta))$. Furthermore, the reduction in NP oracle calls led to significant runtime improvement in practice. It is worth commenting that the usage of {\ApproxMCTwo} as FPRAS for DNF is shown to achieve runtime efficiency over the alternatives based on Monte Carlo methods~\cite{MSV17,MSV18,MSV19}.



\subsection{Minimum-based Algorithm}\label{sec:minimum}
For a given multiset $\mathbf{a}$
( eg: a data stream or solutions to a model), we now specify the property $\mathcal{P}_2(\mathcal{S}, H, \mathbf{a}_{u})$. 
The sketch $\mathcal{S}$ is an array of sets indexed by members of $H$ that holds lexicographically $p$ minimum elements of $H[i](\mathbf{a}_u)$ where $p$ is $\min(\frac{96}{\varepsilon^2}, |\mathbf{a}_{u}| )$.  $\mathcal{P}_2$ is the property that specifies this relationship.  More formally, the relationship $\mathcal{P}_2$ holds, if the following conditions are met.
\begin{enumerate}
	\item $\forall i, |\mathcal{S}[i]|  = \min(\frac{96}{\varepsilon^2}, |\mathbf{a}_{u}| )$
	\item $\forall i, \forall y \notin \mathcal{S}[i],  \forall y' \in \mathcal{S}[i] \text{ it holds that } H[i](y') \preceq H[i](y)$
\end{enumerate}

The following lemma due to Bar-Yossef {\it et al.}~\cite{BJKST02} establishes the relationship between the property $\mathcal{P}_2$ and the number of distinct elements of a multiset. Let $\max(S_i)$ denote the largest element of the set $S_i$.

\begin{lemma}\cite{BJKST02}\label{lem:bjkstmin}
Let $\mathbf{a} \subseteq \{0,1\}^n$ be a multiset and $H \subseteq \Hteop(n,m)$ where $m=3n$ and each $H[i]$s  are independently drawn and $|H|=O(\log 1/\delta)$ and let $\mathcal{S}$ be such that the $\mathcal{P}_2 (\mathcal{S}, H, a_u)$ holds. Let $c = \mbox{\rm Median }\{{p\cdot 2^m \over \max (S[i])}\}_i$. Then 
\[\Pr \left[ \frac{|\mathbf{a}_u|}{(1+\varepsilon)} \leq c
\leq (1+\varepsilon)|\mathbf{a}_u|\right]\geq 1- \delta.\]
\end{lemma}

Therefore, we can transform the \MinValue\ algorithm for $F_0$ estimation to that of model counting given access to a subroutine that can compute $\mathcal{S}$ such that $\mathcal{P}_2(\mathcal{S}, H, \satisfying{\varphi})$ holds true. The following proposition establishes the existence and complexity of such a  subroutine, called {\FindMin}:

\begin{proposition}\label{prop:findmin}
There is an algorithm {\FindMin} that, given $\varphi$ over $n$ variables, $h \in \Hteop(n,m)$, and $p$ as input, returns a set, $\mathcal{B} \subseteq h(\satisfying{\varphi})$ so that if   $|h(\satisfying{\varphi})| \leq p$, then $\mathcal{B} =h(\satisfying{\varphi})$, otherwise $\mathcal{B}$ is the $p$ lexicographically minimum elements of $h(\satisfying{\varphi})$.
Moreover, 
if $\varphi$ is a CNF formula, then {\FindMin} makes $\mathcal{O}(p\cdot m)$ calls to an NP oracle, and if $\varphi$ is a DNF formula with $k$ terms, then {\FindMin} takes $\mathcal{O}(m^3 \cdot n \cdot k \cdot p)$ time. 
\end{proposition}

Equipped with Proposition~\ref{prop:findmin}, we are now ready to present the algorithm, called {\ApproxModelCountMin}, for model counting. Since the complexity of {\FindMin} is PTIME when $\varphi$ is in DNF, we have {\ApproxModelCountMin} as a FPRAS for DNF formulas. 

\begin{theorem}\label{lm:approxmin}
	Given $\varphi$, $\varepsilon$,$ \delta$, {\ApproxModelCountMin} returns $c$ such that 
	\[\Pr \left( \frac{|\satisfying{\varphi}}{1+\varepsilon} \leq Est \leq (1+\varepsilon)|\satisfying{\varphi}|\right) \geq 1- \delta.\]
	  If $\varphi$ is a CNF formula, then {\ApproxModelCountMin} is a polynomial-time algorithm that makes $\mathcal{O}(\frac{1}{\varepsilon^2}  n \log (\frac{1}{\delta}))$ calls to NP oracle. If $\varphi$ is a DNF formula, then {\ApproxModelCountMin} is an FPRAS. 
\end{theorem}

\begin{algorithm}
		\caption{$\ApproxModelCountMin (\varphi,\varepsilon, \delta)$}
		\label{alg:approxmin}
		\begin{algorithmic}[1]
			\State $t \gets 35\log (1/\delta)$
			\State $H\gets \PickHashFunctions (\Hteop(n,3n),t)$
			\State $S \gets \{\}$
			\State $\thresh \gets \frac{96}{\varepsilon^2}$
			\For {$i \in  [1, t] $}
				\State $S[i] \gets \FindMin(\varphi, H[i], \thresh)$
			\EndFor	
			\State $Est \gets {\rm Median}\left( \left\{ \frac{\thresh \times 2^{3n}}{\max\{S[i]\} } \right\}_{i} \right)$
			\State \Return $Est$
		\end{algorithmic}
\end{algorithm}

\subsubsection{Implementing the Min-based Algorithm}

We now give a proof of Proposition~\ref{prop:findmin} by giving an implementation of  {\FindMin} subroutine.

\begin{proof}

We first present the algorithm when the formula $\varphi$ is a DNF formula. Adapting the algorithm for the case of CNF can be done by suing similar ideas. 

Let $\phi = T_1\vee T_2\vee\ \cdots \vee T_k$ be a DNF formula over
$n$ variables where $T_i$ is a term. Let $h:\{0,1\}^n\rightarrow
\{0,1\}^{m}$  be a linear hash function in $\mathcal{H}_{xor}(n,m)$
defined by a $m\times n$ binary matrix $A$.  Let $\mathcal{C}$ be the set of hashed values of the
satisfying assignments for $\varphi$: $\mathcal{C} = \{h(x) \mid x \models \varphi\}
\subseteq \{0,1\}^m$. Let $\mathcal{C}_{p}$ be the first $p$ elements of $\mathcal{C}$ in the lexicographic order. Our goal is to compute $\mathcal{C}_{p}$. 


We will give an algorithm with running time $O(m^3np)$ 
to compute $\mathcal{C}_p$ when the formula is just a
term $T$. Using this algorithm we can compute $\mathcal{C}_p$
for a formula with $k$ terms by iteratively merging $\mathcal{C}_p$ for each term. The time
complexity increases by a factor of $k$, resulting in an $O(m^3nkp)$ time algorithm.

Let $T$ be a term with width $w$ (number of literals) and $\mathcal{C} = \{Ax
\mid x \models T\}$. By fixing the variables in $T$ we get a vector $b_T$ and an $N\times
(n-w)$ matrix $A_T$ so that $\mathcal{C} = \{A_Tx + b_T \mid x\in
\{0,1\}^{(n-w)}\}$. Both $A_T$ and $b_T$ can be computed
from $A$ and $T$ in linear time. Let  $h_T(x)$ be the transformation $A_Tx + b_T$. 

We will compute $\mathcal{C}_p$ ($p$ lexicographically minimum elements in $\mathcal{C}$) iteratively as follows: assuming we have computed $(q-1)^{th}$ minimum of $\mathcal{C}$, we will compute $q^{th}$ minimum using a prefix-searching strategy. We will use  a subroutine to solve the following basic prefix-searching primitive: Given any $l$ bit string $y_1\ldots y_l$, is there an $x\in \{0,1\}^{n-w}$ so that $y_1\ldots y_l$ is a prefix for some string in $\{h_T(x)\}$? This task can be performed using Gaussian elimination over an $(l+1)\times (n-w)$ binary matrix and can be implemented in time $O(l^2(n-w))$.  

Let $y=y_1\ldots y_m$ be the $(q-1)^{th}$ minimum in $\mathcal{C}$. Let $r_1$ be the rightmost 0 of $y$. Then using the above mentioned procedure we can find the lexicographically smallest string in the range of $h_T$ that extends $y_1\ldots y_{(r-1)}1$ if it exists. If no such string exists in $\mathcal{C}$,  find the index of the next   0 in $y$ and repeat the procedure.  In this manner the $q^{th}$ minimum can be computed using $O(m)$ calls to the prefix-searching primitive resulting in an $O(m^3n)$ time algorithm. Invoking the above procedure $p$ times results in an algorithm to compute $\mathcal{C}_p$ in $O(m^3np)$ time.  

If $\varphi$ is a CNF formula, we can employ the same prefix searching strategy. Consider the following NP oracle: $O=\{\langle \varphi, h, y, y' \rangle \mid \exists x, \exists y'', \mbox{ so that } x \models \varphi, y'y'' > y,  h(x) = y'y'' \}$. With $m$ calls to $O$, we can compute string in $\mathcal{C}$ that is lexicographically greater than $y$. So with $p\cdot m$, calls to $O$, we can compute $\mathcal{C}_p$. 

\end{proof}

\subsubsection*{Further Optimizations}
As mentioned in Section~\ref{sec:introduction}, the problem of model counting has witnessed a significant interest from practitioners owing to its practical usages and the recent developments have been fueled by the breakthrough progress in the SAT solving wherein calls to NP oracles are replaced by invocations of SAT solver in practice. Motivated by the progress in SAT solving, there has been significant interest in design of efficient algorithmic frameworks for related problems such as MaxSAT and its variants. The state of the art MaxSAT based on sophisticated strategies such as implicit hitting sets and are shown to significant outperform algorithms based on merely invoking a SAT solver iteratively.  Of particular interest to us is the recent progress in the design of MaxSAT solvers to handle lexicographic objective functions. In this context, it is worth remarking that we expect practical implementation of {\FindMin} would invoke a MaxSAT solver $\mathcal{O}(p)$ times.


\subsection{Estimation-based Algorithm}\label{sec:estim}

We now adapt the \Estimation\ algorithm to model counting. For a given stream $\mathbf{a}$ and chosen hash functions $H$,  the sketch $\mathcal{S}$ corresponding to the estimation-based algorithm satisfies the following relation $\mathcal{P}_3(\mathcal{S}, H, \mathbf{a}_u)$:
 
 \begin{align}
 \mathcal{P}_{3}(\mathcal{S}, H, \mathbf{a}_u) := \left(S[i,j] = \max_{x \in \mathbf{a}_{u}} \TrailZero(H[i,j])(x)\right)
 \end{align}
where the procedure $\TrailZero(z)$ is the length of the longest all-zero suffix of $z$. Bar-Yossef {\it et al.} \cite{BJKST02} show the following relationship between the property $\mathcal{P}_3$ and $F_0$.

\begin{lemma}\cite{BJKST02}
Let $\mathbf{a} \subseteq \{0,1\}^n$ be a multiset. For $i \in [T]$ and $j \in [M]$, suppose $H[i,j]$ is drawn independently from $\mathcal{H}_{s-{\rm wise}}(n,n)$ where $s = O(\log(1/\varepsilon))$, $T = O(\log(1/\delta))$, and $M = O(1/\varepsilon^2)$. Let $H$ denote the collection of these hash functions. Suppose $\mathcal{S}$ satisfies $\mathcal{P}_3(\mathcal{S}, H, \mathbf{a}_u)$. For any integer $r$, define:
\[
c_r = \mathrm{Median}\left\{\frac{\ln\left(1-\frac{1}{M}\sum_{j=1}^M \mathbb{1}\{\mathcal{S}[i,j]\geq r\}\right)}{\ln(1-2^{-r})}\right\}_i
\]
Then, if $2F_0 \leq 2^r \leq 50F_0$:
\[
\Pr\left[(1-\varepsilon)F_0 \leq c_r \leq (1+\varepsilon)F_0\right] \geq 1-\delta
\]
\end{lemma}

Following the recipe outlined above, we can transform a $F_0$ streaming algorithm to a model counting algorithm by designing a subroutine that can compute the sketch for the set of all solutions described by $\varphi$ and a subroutine to find $r$. The following proposition achieves the first objective for CNF formulas using a small number of calls to an NP oracle:

\begin{proposition}\label{prop:range}
There is an algorithm	{\FindMaxRange}  that given $\varphi$ over $n$ variables and hash function $h \in \mathcal{H}_{s-{\rm wise}}(n,n)$, returns $t$ such that 
	\begin{enumerate}
		\item $\exists z, z \models \varphi$ and $h(z)$ has $t$ least significant bits equal to zero. 
		\item $\forall (z \models \varphi) \implies $ $h(z)$ has $ \leq t$ least significant bits equal to zero. 
	\end{enumerate}
	If $\varphi$ is a CNF formula, then {\FindMaxRange} makes $\mathcal{O}(\log n)$ calls to an NP oracle.
\end{proposition}
\begin{proof}
Consider an NP oracle $O= \{\langle \varphi, h, t\rangle \mid \exists x, \exists y, x \models \varphi, h(x) = y0^t\rangle\}$. Note that $h$ can be implemented as a degree-$s$ polynomial $h: \mathbb{F}_{2^n} \to \mathbb{F}_{2^n}$, so that $h(x)$ can be evaluated in polynomial time. A binary search, requiring $O(\log n)$ calls to $O$, suffices to find the largest value of $t$ for which $\langle \varphi, h, t\rangle$ belongs to $O$.
\end{proof}
We note that unlike Propositions \ref{prop:boundedsat} and \ref{prop:findmin}, we do not know whether {\FindMaxRange} can be implemented efficiently when $\varphi$ is a DNF formula.  For a degree-$s$ polynomial $h: \mathbb{F}_{2^n} \to \mathbb{F}_{2^n}$, we can efficiently test whether $h$ has a root by computing $\mathsf{gcd}(h(x), x^{2^n}-x)$, but it is not clear how to simultaneously constrain some variables according to a DNF term.

Equipped with Proposition~\ref{prop:range}, we obtain {\ApproxModelCountEst} that takes in a formula $\varphi$ and a suitable value of $r$  and returns $|\satisfying{\varphi}|$. The key idea of {\ApproxModelCountEst} is to repeatedly invoke {\FindMaxRange} for each of the chosen hash functions and compute the estimate based on the sketch $\mathcal{S}$ and the value of $r$.  The following lemma summarizes the time complexity and guarantees of {\ApproxModelCountEst} for CNF formulas.
 
 \begin{theorem}\label{lm:est}
 	Given a CNF formula $\varphi$, parameters $\varepsilon$ and $ \delta$, and $r$ such that $2F_0 \leq 2^r \leq 50F_0$, the algorithm {\ApproxModelCountEst} returns $c$ satisfying \[\Pr \left[ \frac{|\satisfying{\varphi}}{1+\varepsilon} \leq c \leq (1+\varepsilon)|\satisfying{\varphi}|\right] \geq 1- \delta.\]
 	 {\ApproxModelCountEst} makes $\mathcal{O}(\frac{1}{\varepsilon^2} \log n \log (\frac{1}{\delta}))$ calls to an NP oracle.
 \end{theorem}

In order to obtain $r$, we run in parallel another counting algorithm based on the simple 
$F_0$-estimation algorithm \cite{FM85, AMS99} which we call 
$\mathsf{FlajoletMartin}$.
Given a stream $\mathbf{a}$, the $\mathsf{FlajoletMartin}$ algorithm chooses a random pairwise-independent hash function $h \in H_{xor}(n,n)$, computes the largest $r$ so that for some  $x \in \mathbf{a}_u$,  the $r$ least significant bits of $h(x)$ are zero, and outputs $r$. Alon, Matias and Szegedy \cite{AMS99} showed that $2^r$ is a 5-factor approximation of $F_0$ with probability $3/5$. Using our recipe, we can convert $\mathsf{FlajoletMartin}$ into an algorithm that approximates the number of solutions to a CNF formula $\varphi$ within a factor of 5 with probability 3/5. It is easy to check that using the same idea as in Proposition \ref{prop:range}, this algorithm requires $O(\log n)$ calls to an NP oracle.

\begin{algorithm}
	\caption{$\ApproxModelCountEst(\varphi,\varepsilon, \delta,r)$}
	\label{alg:approxest}
	\begin{algorithmic}[1]
		\State $\thresh \gets 96/\varepsilon^2$
		\State $t \gets 35 \log(1/\delta)$
		\State $H \gets \PickHashFunctions(\mathcal{H}_{s-{\rm wise}}(n,n), t \times \thresh)$
		\State $S \gets \{\}$
		\For {$i \in  [1, t] $}
		\For {$j \in [1, \thresh]$}
		\State $S[i,j] \gets \FindMaxRange(\varphi, \TrailZero(H[i,j]))$
		\EndFor
		\EndFor	
		\State $Est \gets \mathrm{Median}\left\{\frac{\ln\left(1-\frac{1}{\thresh}\sum_{j=1}^\thresh \mathbb{1}\{\mathcal{S}[i,j]\geq r\}\right)}{\ln(1-2^{-r})}\right\}_i$
		\State \Return $Est$
	\end{algorithmic}
\end{algorithm}

\subsection{The Opportunities Ahead}
As noted in Section~\ref{sec:bucketing}, the algorithms based on Bucketing was already known
and have witnessed a detailed technical development from both applied and algorithmic perspectives. The 
model counting algorithms based on  Minimum and Estimation are new. We discuss some 
potential implications of these  new algorithms to SAT solvers and other aspects. 

\paragraph{MaxSAT solvers with native support for XOR constraints} 
When the input formula $\varphi$ is represented as CNF, then  {\ApproxMC}, the model counting
	algorithm based on Bucketing strategy, invokes NP oracle over CNF-XOR formulas, i.e., formulas 
	expressed as  conjunction of CNF and XOR constraints. The significant improvement in runtime
	performance of {\ApproxMC} owes to the design of SAT solvers with native support for CNF-XOR formulas~\cite{SNC09,SM19,SGM20}.
	Such solvers have now found applications in other domains such as cryptanalysis. 
	It is perhaps worth emphasizing that the proposal of {\ApproxMC} was crucial to renewed interest in 
	the design of SAT solvers with native support for CNF-XOR formulas.   
	As observed in Section~\ref{sec:minimum}, the algorithm based on Minimum strategy would ideally 
	invoke a MaxSAT solver that can handle XOR constraints natively. We believe that Minimum-based algorithm
	will ignite interest in the design of MaxSAT solver with native support for XOR constraints. 
	
\paragraph{FPRAS for DNF based on Estimation}
 In Section~\ref{sec:estim}, we were unable to  show that the
 model counting algorithm obtained based on Estimation is FPRAS when $\varphi$ is represented as DNF. 
  The algorithms based on Estimation have been shown to achieve optimal space efficiency in the context of 
 $F_0$ estimation. In this context, an open problem is to investigate whether Estimation-based strategy lends itself
 to FPRAS for DNF counting. 
 
 \paragraph{Empirical Study of FPRAS for DNF Based on Minimum}
 Meel et al.~\cite{MSV18,MSV19} observed that FPRAS for DNF based on Bucketing has superior performance, in terms of the number of instances solved, to that of FPRAS schemes
 based on Monte Carlo framework. In this context, a natural direction of future work would be to conduct an empirical study to  
 understand behavior of FPRAS scheme based on Minimum strategy.


\section{Distributed DNF Counting} \label{sec:distributed}

Consider the problem of {\em distributed DNF counting}. In this setting, there are $k$ sites that can each communicate with a central coordinator. The input DNF formula $\varphi$ is partitioned into $k$ DNF subformulas $\varphi_1, \dots, \varphi_k$, where each $\varphi_i$ is a subset of the terms of the original $\varphi$, with the $j$'th site receiving only $\varphi_j$. The goal is for the coordinator to obtain an $(\epsilon,\delta)$-approximation of the number of solutions to $\varphi$, while minimizing the total number of bits communicated between the sites and the coordinator. Distributed algorithms for sampling and counting solutions to CSP's have been studied recently in other models of distributed computation \cite{FSY18, FHY18, FY18, FG18}. From a practical perspective, given the centrality of \#DNF in the context of probabilistic databases~\cite{S18,RS08}, a distributed DNF counting would entail applications in  distributed probabilistic databases.  

From our perspective, distributed DNF counting falls within the {\em distributed functional monitoring} framework formalized by Cormode et al. \cite{CMY11}. Here, the input is a stream $\mathbf{a}$ which is partitioned arbitrarily into sub-streams $\mathbf{a}_1, \dots, \mathbf{a}_k$ that arrive at each of $k$ sites. Each site can communicate with the central coordinator, and the goal is for the coordinator to compute a function of the joint stream $\mathbb{a}$ while minimizing the total communication. This general framework has several direct applications and has been studied extensively \cite{BO03, CGMR05, MSDO05, KCR06, ABC09, SSK10, SSK11, CMYZ12, WZ12, YZ13, HYZ12, WZ17}. In distributed DNF counting, each sub-stream $\mathbf{a}_i$ corresponds to the set of satisfying assignments to each subformula $\varphi_i$, while the function to be computed is $F_0$. 

The model counting algorithms discussed in Section \ref{sec:counting} can be extended to the distributed setting. We briefly indicate the distributed implementations for each of the three algorithms. As above, we set the parameters $\thresh$ to $O(1/\varepsilon^2)$.  Correctness follows from Bar-Yossef {\it et al.} \cite{BJKST02} and the earlier discussion. 

We consider an adaptation of {\BoundedSAT} that takes in $\varphi$ over $n$ variables,  a hash function $h \in \Hteop (n, m)$, and a threshold $t$ as inputs, returns a set $U$ of solutions such that $|U| = \min (t, |\satisfying{\varphi \wedge h(x) = {0^m}}|)$. 

\begin{description}
\item[Bucketing.] Setting $m=$ $O(\log(k/\delta\varepsilon^2))$, the coordinator chooses $H[1], \dots, H[T]$ from $\Hteop(n,n)$ and $G$ from $\Hxor(n,m)$. It then sends them to the $k$ sites.  Let $m_{i,j}$ be the  smallest  $m$ such that the size of the set $\BoundedSAT(\varphi_j, H[i]_m, \mathsf{thresh})$ is smaller than $\mathsf{thresh}$.   The $j$'th site sends the coordinator the following tuples: $\langle G(x), \mathsf{TrailZero}(H[i](x))\rangle$ for each $i \in [t]$ and for each $x$ in $\BoundedSAT(\varphi_j, H[i]_{m_{i,j}},$ $\mathsf{thresh})$. Note that each site only sends tuples for at most $O(1/\delta \varepsilon^2)$ choices of $x$, so that $G$ hashes these $x$ to distinct values with probability $1-\delta/2$. It is easy to verify that the coordinator can then execute the rest of the algorithm $\mathsf{ApproxMC}$.  The communication cost is $\tilde{O}(k(n+1/\varepsilon^2) \cdot \log(1/\delta))$, and the time complexity for each site is polynomial in $n$, $\varepsilon^{-1}$, and $\log(\delta^{-1})$.

\item[Minimum.] The coordinator chooses hash functions $H[1],\dots,H[t]$ from $\Hteop(n,3n)$ and sends it to the $k$ sites. Each site runs the $\mathsf{FindMin}$ algorithm for each hash function and sends the outputs to the coordinator. So, the coordinator receives sets $S[i,j]$, consisting of the $\thresh$ lexicographically smallest hash values of the solutions to $\varphi_j$. The coordinator then extracts $S[i]$, the $\thresh$ lexicographically smallest elements of $S[i,1] \cup \dots \cup S[i,k]$ and proceeds with the rest of algorithm $\mathsf{ApproxModelCountMin}$. The communication cost is $O(kn/\varepsilon^2 \cdot \log(1/\delta))$ to account for the $k$ sites sending the outputs of their $\mathsf{FindMin}$ invocations.  The time complexity for each site is polynomial in $n$, $\varepsilon^{-1}$, and $\log(\delta^{-1})$. 

\item[Estimation.] For each $i \in [t]$, the coordinator chooses $\thresh$ hash functions $H[i,1], \dots, H[i,\thresh]$, drawn pairwise independently from $\mathcal{H}_{s-{\rm wise}}(n,n)$ (for $s = O(\log(1/\varepsilon))$) and sends it to the $k$ sites. Each site runs the $\mathsf{FindMaxRange}$ algorithm for each hash function and sends the output to the coordinator. Suppose the coordinator receives $S[i,j, \ell] \in [n]$ for each $i \in [t], j \in [\thresh]$ and $\ell \in [k]$. It computes $S[i,j] = \max_\ell S[i,j,\ell]$. The rest of $\mathsf{ApproxModelCountEst}$ is then executed by the coordinator. The communication cost is $\tilde{O}(k(n+1/\varepsilon^2)\log(1/\delta))$. However, as earlier, we do not know a polynomial time algorithm to implement the $\mathsf{FindMaxRange}$ algorithm for DNF terms.
\end{description}

\subsection*{Lower Bound}
The communication cost for the Bucketing- and Estimation-based algorithms is nearly optimal in their dependence on $k$ and $\varepsilon$. Woodruff and Zhang \cite{WZ12} showed that the randomized communication complexity of estimating $F_0$ up to a $1+\varepsilon$ factor in the distributed functional monitoring setting is $\Omega(k/\varepsilon^2)$. We can reduce $F_0$ estimation problem to distributed DNF counting. Namely, if for the $F_0$ estimation problem, the $j$'th site receives items $a_1, \dots, a_m \in [N]$, then  for the distributed DNF counting problem,  $\varphi_j$ is a DNF formula on $\lceil \log_2 N \rceil$ variables whose solutions are exactly $a_1, \dots, a_m$ in their binary encoding. Thus, we immediately get an $\Omega(k/\varepsilon^2)$ lower bound for the distributed DNF counting problem. Finding the optimal dependence on $N$ for $k>1$ remains an interesting open question\footnote{Note that if $k=1$, then $\log(n/\varepsilon)$ bits suffices, as the site can solve the problem on its own and send to the coordinator the binary encoding of a $(1+\varepsilon)$-approximation of $F_0$.}.
\section{From Counting to Streaming: Structured Set Streaming}\label{sec:streaming}

\ignore{
We consider a universe $U= \{0,1\}^n$. For any $x\in U$, we interpret $x$ as both as an $n$-bit string 
as well as encoding an integer between $0$ and $2^n-1$. We denote $|U| = N = 2^n$. }

In this section we consider  {\em structured set streaming model} where each item $S_i$ of the stream is a succinct  representation of a set over the universe $U = \{0,1\}^n$.  Our goal is to design efficient algorithms (both in terms of memory and processing time per item)  for computing $|\cup_i S_i|$ - number of distinct elements in the union of all the sets in the stream. We call this problem $F_0$ computation over structured set streams.

\subsection*{DNF Sets}
A particular representation we are interested in is where each set is presented as the set of satisfying assignments to a  DNF formula. Let $\varphi$ is a DNF formula over $n$ variables. Then the {\em DNF Set} corresponding to $\varphi$ is the set of satisfying assignments of $\varphi$. The {\em size} of this representation is  the number of terms in the formula $\varphi$. 

A stream over DNF sets is a stream of DNF formulas $\varphi_1, \varphi_2, \ldots$.
Given such a DNF stream, the goal is to estimate  $|\bigcup_{i} S_i|$ where  $S_i$ the DNF set represented by $\varphi_i$. This quantity is same as the number of satisfying assignments of the formula $\vee_i \varphi_i$. We show that the algorithms described in the previous section carry over to obtain $(\epsilon, \delta)$ estimation algorithms for this problem with space and per-item time $\mathrm{poly}(1/\epsilon,  n, k, \log (1/\delta))$ where $k$ is the size of the formula. 

Notice that this model generalizes the traditional streaming model where each item of the stream is an element $x\in U$ as it can be represented as single term DNF formula $\phi_x$ whose only satisfying assignment is $x$. This model also subsumes certain other models considered in the streaming literature that we discuss later.

\begin{theorem}\label{DNFStream:thm}
There is a streaming algorithm to compute an $(\epsilon, \delta)$ approximation of $F_0$  over DNF sets. This algorithm takes space $O({n\over \varepsilon^2}\cdot\log {1\over \delta})$ and processing time $O(n^4\cdot k\cdot{1\over \varepsilon^2}\cdot \log{1\over \delta})$ per item where $k$ is the size (number of terms) of the corresponding DNF formula. 

\end{theorem}

\begin{proof}
We show how to adapt  Minimum-value based algorithm from Section~\ref{sec:minimum} to this setting. The algorithm picks a hash function $h \in \Hteop(n,3n)$ maintains the set $\mathcal{B}$ consisting of  $t$ lexicographically minimum elements of the set $\{h(\satisfying{\varphi_1, \vee \ldots \vee, \varphi_{i-1}})\}$ after processing $i-1$ items. When $\varphi_i$ arrives, it computes the set $\mathcal{B'}$
consisting of the  $t$ lexicographically minimum values of the set $\{h(\satisfying{\varphi_i})\}$ and subsequently updates $\mathcal{B}$ by computing the $t$ lexicographically smallest elements from 
$\mathcal{B}\cup \mathcal{B'}$. By Proposition~\ref{prop:findmin}, computation of $\mathcal{B'}$ can be done in time $O(n^4\cdot k \cdot t)$ where $k$ is the number of terms in $\varphi_i$. Updating $\mathcal{B}$ can be done in $O(t\cdot n)$ time. Thus update time for the item  $\varphi_i$ is $O(n^4 \cdot k \cdot t)$. For obtaining an $(\varepsilon, \delta)$ approximations we set 
$t = O({1\over \varepsilon^2})$ and repeat the procedure $O(\log {1\over \delta})$ times and take the median value. Thus the update time for item $\varphi$ is $O(n^4\cdot k\cdot{1\over \varepsilon^2}\cdot \log{1\over \delta})$. For analyzing sapce, each hash function uses $O(n)$ bits and to store $O({1 \over \epsilon^2})$ minimums, we require $O({n \over \epsilon^2})$ space resulting in overall space usage of $O({n\over \varepsilon^2}\cdot\log {1\over \delta})$. The proof of correctness follows from Lemma~\ref{lem:bjkstmin}. 
\end{proof}

Instead of using  Minimum-value based algorithm, we could adapt Bucketing-based algorithm to obtain an algorithm with similar space and time complexities. 
As noted earlier, some of  the set streaming models considered in the literature can be reduced the DNF set streaming. We discuss them next.

\subsection*{Multidimensional Ranges}\label{sec:multidimrange}

A $d$ dimensional range over an universe $U$ is defined as  $[a_1,b_1] \times [a_2,b_2] \times \ldots \times [a_d, b_d]$. Such a range  represents the set of tuples $\{(x_1,\ldots,x_d)$ where $a_i \leq x_i \leq b_i$  and $x_i$ is an integer. Note that  every $d$-dimensional range can be succinctly by the tuple $\langle a_1, b_1, \cdots a_d, b_d\rangle$. A multi-dimensional stream is a stream where each item is a $d$-dimensional range. The goal is to compute $F_0$ of the union of  the $d$-dimensional ranges effeiciently. We will show that $F_0$ computation over multi-dimensional ranges can reduced to $F_0$ computation over DNF sets. Using this reduction we arrive at a simple algorithm to compute $F_0$ over multi-dimensional ranges.


\begin{lemma}\label{rangeToDNF:lem}
Any $d$-dimensional range $R$ over $U$ can be represented as a DNF formula $\varphi_R$  over $nd$ variables whose size is at most $(2n)^d$. There is algorithm that takes $R$ as input and outputs the $i^{th}$ term of $\varphi_R$ using $O(nd)$  space, for $1 \leq i \leq (2n)^d$.  
\end{lemma}

\begin{proof}
Let $R = [a_1,b_1] \times [a_2,b_2] \times \ldots \times [a_d, b_d]$ be a $d$-dimensional range over $U^{d}$.
We will first  describe the  formula to represent the multi-dimensional range as a conjunction of $d$ DNF formulae $\phi_1, \cdots, \phi_d$ each with at most $2n$ terms, where $\phi_i$ represents $[a_i,b_i]$, the range in the $i^{th}$ dimension.  Converting this into a DNF formula will result in the formula  $\phi_R$ with $(2n)^d$ terms.

For any $\ell$ bit number $c$, $1 \leq c \leq 2^n$,  it is straightforward to write a DNF formula $\varphi_{\leq c}$, of size at most $\ell$,  that represents the range $[0,c]$ (or equivalently the set $\{x\mid 0\leq x \leq c\}$). Similarly we can write  a DNF formula $\varphi_{\geq c}$, of size at most $\ell$ for the range $[c,2^{\ell-1}]$. Now we construct a formula to represent the range $[a,b]$ over $U$ as follows. Let $a_1a_{2}\cdots a_n$ and $b_1b_{2}\cdots b_n$ be the binary representations of $a$ and $b$ respectively.  Let $\ell$ be the largest integer such that $a_1a_2\cdots a_l = b_1b_2\cdots b_l$. Hence $a_{\ell+1} = 0$ and $b_{\ell+1} = 1$.
Let $a'$ and $b'$ denote the integers represented by $a_{l+2}\cdots a_n$ and $b_{l+2} \cdots b_n$. Also, let $\psi$ denote the formula (a single term) that represents the string $a_1\cdots a_\ell$. Then the formula representing $[a,b]$ is $\psi \wedge (\overline{x_{\ell+1}} \varphi_{\geq a'} \vee x_{\ell+1}\varphi_{\leq b'})$. This can be written as a DNF formula by distributing $\psi$ and the number of terms in the resulting formula is at most $2n$, and has  $n$ variables. 
Note that each $\varphi_i$ can be constructed using $O(n)$ space.  To obtain the final DNF representing the range $R$, we need to convert $\varphi_1 \wedge \cdots \varphi_d$ into a DNF formula. It is easy to see that for any $i$, then $i$th term of this DNF can be computed using space $O(nd)$.  Note that this formula has $nd$ variables, $n$ variables per each dimension.


\end{proof}

Using the above reduction  and Theorem~\ref{DNFStream:thm}, we obtain an  an algorithm for estimating $F_0$ over multidimensional ranges in a range-efficient manner.

\begin{theorem}\label{RangeEfficient:cor}
There is a streaming algorithm to compute an $(\epsilon, \delta)$ approximation of $F_0$  over $d$-dimensional ranges that takes space $O( \frac{nd}{\varepsilon^2}\cdot \log (1/\delta))$ and processing time $O((nd)^4\cdot n^d \cdot \frac{1}{\varepsilon^2})\log(1/\delta))$ per item. 

\end{theorem}

\begin{remark}\label{remark}
Tirthapura and Woodruff~\cite{TW12} studied the problem of range efficient estimation of $F_k$ ($k^{th}$ frequency moments) over $d$-dimensional ranges. They claimed an algorithm to estimate $F_0$ with space  and per-item time complexity $\mathrm{poly}(n, d, 1/\epsilon,\log 1/\delta)$.  However they have retracted their claim~\cite{Woodruff20}. Their method only yields $\mathrm{poly}(n^d,1/\epsilon,\log 1/\delta)$ time per item.
 Their proof appears to be involved that require a range efficient implementations of {\em count sketch} algorithm~\cite{CCF04}  and recursive sketches~\cite{IW05,BO10}. We obtain the same complexity bounds with much simpler analysis and a practically efficient algorithm that can use off the shelf available implementations~\cite{MSV18}. 
 \end{remark}

\begin{remark}
	Subsequent to the present work, 
	an improved algorithm for $F_0$ over  structured sets is presented in ~\cite{CMV21} (to appear in PODS 2021). In particular, the paper presents an $F_0$ estimation algorithm, called {\apsestimator},  for streams over {\em Delphic sets}.  
	A set $S \subseteq \{0,1\}^n$ belongs to Delphic family if the following queries can be done in $O(n)$ time: (1) know the size of the set $S$, (2) draw a uniform random sample from $S$, and
	(3) given any $x$ check if $x\in S$. The authors design a streaming algorithm  that given a 
	stream $\mathcal{S}= \langle S_1, S_2 \cdots, S_M \rangle$  wherein each $S_i \subseteq \{0,1\}^n$ belongs to Delphic family,  computes an $(\varepsilon,\delta)$-approximation of $| \bigcup_{i=1}^{M} S_i|$ with  worst case space complexity  
	$O(n\cdot\log (M/\delta)\cdot \varepsilon^{-2})$ and per-item time is $\widetilde{O}(n \cdot \log (M/\delta)\cdot \varepsilon^{-2})$. The algorithm \apsestimator, when applied to $d$-dimensional ranges, gives per-item time and space complexity bounds that are $\mathrm{poly}(n, d, \log M, 1/\varepsilon, \log 1/\delta)$.  While {\apsestimator} brings down the dependency on $d$ from exponential to polynomial, it works under the assumption that the length of the stream $M$ is known. The general setup presented in~\cite{CMV21}, however, can be applied to other structured sets considered in this paper including multidimensional arithmetic progressions. 	
\end{remark}

\paragraph{Representing Multidimensional Ranges as CNF Formulas.} 
Since the algorithm, {\apsestimator}, presented in~\cite{CMV21}, employs a sampling-based technique, a natural question is whether there exists a hashing-based technique that achieves per-item time polynomial in $n$ and $d$. We note that the above approach of representing a multi-dimensional range as DNF formula does not yield such an algorithm. This is because there exist $d$-dimensional ranges whose DNF representation requires $\Omega(n^d)$ size.


\begin{observation}
There exist $d$-dimensional ranges whose DNF representation has size $\geq n^d$.
\end{observation}
\begin{proof}
The observation follows by considering the range $R = [1,2^n-1]^d$ (only 0 is missing from the interval in each dimension). We will argue that any DNF formula $\varphi$ for this range has size (number of terms) $ \geq n^d$.
For any $1\leq j \leq d$, we use the set of variables $X^{j}=\{x^j_1,x^j_2,\ldots,x^j_n\}$ for representing the $j^{th}$ coordinate of $R$. Then $R$ can be represented as the formula $\varphi_{R} = \vee_{(i_1,i_2,\ldots,i_d)} x_{i_1}^1x_{i_2}^2\ldots x_{i_d}^d$, where $1\leq i_j \leq n$. This formula has $n^d$ terms. Let $\varphi$ be any other DNF formula representing $R$. The main observation is that any term $T$ of $\varphi$ is completely contained (in terms of the set of solutions) in one of the terms of $\varphi_{R}$. This implies that $\varphi$ should have $n^d$ terms. Now we argue that $T$ is contained in one of the terms of $\varphi_{R}$. $T$ should have at least one variable as positive literal from each of $X^j$. Suppose  $T$  does not have any variable from $X^j$ for some $j$. Then $T$ contains a solution with all the variable in $X^j$ set to 0 and hence not in $R$. Now let $x^j_{i_j}$ be a variable from  $X^j$ that is in $T$. Then clearly $T$ is in the term $x^1_{i_1}x^2_{i_2}\ldots x^d_{i_d}$ of $R$. 
\end{proof}

This leads to the question of whether we can obtain a super-polynomial  lower bound on the time per item. We observe that such a lower bound would imply $\mathrm{P} \neq \mathrm{NP}$. For this, we note the following.

\begin{observation}
Any $d$-dimensional range $R$ can be represented as a CNF formula of size $O(nd)$ over $nd$ variables. 
\end{observation}

This is because a single dimensional range $[a, b]$ can also be represented as a CNF formula of size $O(n)$~\cite{CFMV15} and thus the CNF formula for $R$ is a conjunction of formulas along each dimension. Thus the problem of computing $F_0$ over $d$-dimensional ranges reduces to computing $F_0$ over a stream where each item of the stream is a CNF formula. As in the proof of Theorem~\ref{DNFStream:thm}, we can adapt Minimum-value based algorithm for CNF streams.  When a CNF formula $\varphi_i$ arrive, we need to compute the $t$ lexicographically smallest elements of  $h(\satisfying{\varphi_i})$ where $h \in \Hteop(n,3n)$.  By Proposition~\ref{prop:findmin},  this can be done in polynomial-time by making $O(tnd)$ calls to an NP oracle since $\phi_i$ is a CNF formula over $nd$ variables. Thus if P equals NP, then the time taken per range is polynomial in $n$, $d$, and $1/\varepsilon^2$. Thus a super polynomial time lower bound on time per item implies that P differs from NP.

\subsubsection*{From Weighted \#DNF  to d-Dimensional Ranges}
Designing a streaming algorithm with a  per item of polynomial in $n$ and $d$ is a very interesting open problem with implications on weighted  DNF counting.  Consider a formula $\varphi$ defined on the set of variables $x = \{x_1, x_2, \ldots  x_n  \}$. Let a weight function $\rho: x \mapsto (0,1)$ 
be such that weight of an assignment $\sigma$ can be defined as follows: 
\begin{align*}
W(\sigma) = \prod_{x_i: \sigma(x_i) = 1} \rho(x_i) \prod_{x_i:\sigma(x_i) = 0} (1-\rho(x_i))
\end{align*} 
Furthermore, we define the weight of a formula $\varphi$ as 
\begin{align*}
W(\varphi) = \sum_{\sigma \models \varphi} W(\sigma)
\end{align*}

Given $\varphi$ and $\rho$, the problem of weighted counting is to compute $W(\varphi)$. We consider the case where for each $x_i$, $\rho(x_i)$ is represented using $m_i$ bits in binary representation, i.e., $\rho(x_i) = \frac{k_i}{2^{m_i}}$. Inspired by the key idea of weighted to unweighted reduction due to Chakraborty et al.~\cite{CFMV15}, we show how the problem of weighted  DNF counting can be reduced to that of estimation of $F_0$ estimation of  $n$-dimensional ranges. The reduction is as follows: we transform every term of $\varphi$ into a product of multi-dimension ranges where every variable $x_i$ is replaced with interval $[1,k_i]$ while $\neg x_i$ is replaced with $[k_i+1, 2^{m_i}]$ and every $\wedge$ is replaced with $\times. $
For example, a term $(x_1 \wedge \neg x_2 \wedge \neg x_3)$ is replaced with $[1,k_1]\times[k_2+1, 2^{m_2}]\times[k_3+1,2^{m_3}]$. Given $F_0$ of the resulting stream, we can compute the weight of $\varphi$ simply as $ W(\varphi) = \frac{F_0}{2^{\sum_i m_i} }$.  Thus a hashing based streaming algorithm that has $poly(n, d)$ time per item, yields a hashing based FPRAS for weighted DNF counting, and open problem from~\cite{ACL19}.




\subsubsection{Multidimensional Arithmetic Progressions.}


We will now generalize Theorem~\ref{RangeEfficient:cor}  to handle arithmetic progressions instead of ranges.  Let $[a, b, c]$ represent the arithmetic progression with common difference $c$ in the range $[a, b]$, i.e., $a , a+c, a+2c, a + id$, where $i$ is the largest integer such that $a+ id \leq b$.  
Here,  we consider $d$-dimensional arithmetic progressions  $R = [a_1, b_1, c_1] \times \cdots  \times [a_d, b_d, c_d]$ where each $c_i$ is a power two. We first observe that  the set represented by $[a, b, 2^\ell]$ can be expressed as a DNF formula as follows: Let $\phi$ be the DNF formula representing the range $[a, b]$ and let $a_1, \cdots, a_\ell$ are the least significant bits of $a$, Let $\psi$ be the  term that represents the bit sequence $a_1 \cdots a_\ell$. Now the formula to represent the arithmetic progression $[a, b, 2^\ell]$ is  $\phi \wedge \psi$ which can be converted to a DNF formula of size $O(2n)$. Thus the multi-dimensional arithmetic progression $R$ can be represented as a DNF formula of size $(2n)^d$. Note that time and space required to convert $R$ into a DNF formula are as before, i.e, $O(n^d)$ time and $O(nd)$ space.   This leads us to the following corollary.

\begin{corollary}
There is a streaming algorithm to compute an $(\epsilon, \delta)$ approximation of $F_0$  over $d$-dimensional arithmetic progressions, whose common differences are powers of two, that takes space $O(nd/\varepsilon^2\cdot\log 1/\delta)$ and processing time $O((nd)^4\cdot n^d \cdot \frac{1}{\varepsilon^2})\log(1/\delta))$ per item. 

\end{corollary}

\subsection*{Affine Spaces}
Another example of structured stream is where each item of the stream is an affine space represented by $Ax = B$ where $A$ is a boolean matrix and $B$ is a zero-one vector.   Without loss of generality, we may assume that where $A$ is a $n \times n$  matrix.  Thus  an affine stream consists of $\langle A_1, B_\rangle, \langle A_2, B_2\rangle \cdots$, where each $\langle A_i, B_i\rangle$ is succinctly  represents a set $\{x \in \{0,1 \}^n\mid A_ix= B_i\}$.

For a $n \times n$ Boolean matrix $A$ and a zero-one vector $B$, let $\satisfying{\langle A, B\rangle}$ denote the set of all $x$ that satisfy $Ax = B$.

\begin{proposition}
Given $(A, B)$,  $h \in \Hteop(n,3n)$, and $t$ as input, there is an algorithm, {\AffineFindMin}, that returns a set, $\mathcal{B} \subseteq h(\satisfying{\langle A, B\rangle})$ so that if   $|h(\satisfying{\langle A, B\rangle)}| \leq t$, then $\mathcal{B} =h(\satisfying{\langle A, B\rangle})$, otherwise $\mathcal{B}$ is the $t$ lexicographically minimum elements of $h(\satisfying{\langle A, B\rangle})$. Time taken by this algorithm is $O(n^4t)$ and the space taken the algorithm is $O( tn)$.
\end{proposition}

\begin{proof}
 Let $D$ be the matrix that specifies the hash function $h$.  Let $\mathcal{C} = \{Dx~|~ Ax =B\}$, and the goal is to compute the $t$ smallest element of $\mathcal{C}$.  Note that if $y \in \mathcal{C}$, then it must be the case that $D|Ax =y|B$ where $D|A$ is the matrix obtained by appending rows of $A$ to the rows of $D$ (at the end), and $y|B$ is the vector obtained by appending $B$ to $y$. Note that $D|A$ is a matrix with $4n$ rows.
Now the proof is very similar to the proof of Proposition~\ref{prop:findmin}.  We can do a prefix search as before and this involves doing Gaussian elimination using  sub matrices of $D|A$. 
\end{proof}



\begin{theorem}
There is a streaming algorithms  computes $(\epsilon, \delta)$ approximation of  $F_0$ over affine spaces. This algorithm takes space $O(\frac{n}{\epsilon^2}\cdot \log(1/\delta))$ and processing time of $O(n^4\frac{1}{\epsilon^2}\log(1/\delta))$ per item.
\end{theorem}

\section{Conclusion and Future Outlook}\label{sec:conclusion}


To summarize, our investigation led to a diverse set of results that unify over two decades of work in model counting and $F_0$
estimation. We believe that the viewpoint presented in this work has potential to spur several new interesting research directions. 
We sketch some of these directions below:

\begin{description}
	\item[Sampling] The problem of counting and sampling are closely related. In particular, the seminal work of Jerrum,
	Valiant, and Vazirani~\cite{JVV86} showed that the problem of approximate counting and almost-uniform sampling are inter-reducible for self-reducible
	NP problems. Concurrent to developments in approximate model counting, there has been a significant interest in the design of efficient
	sampling algorithms. A natural direction would be to launch a similar investigation. 
	
	\item[Higher Moments] There has been a long line of work on  estimation of higher moments, i.e. $F_k$ in streaming context. 
	A natural direction of future research is to adapt the notion of $F_k$ in the context of CSP. For example, in the context of DNF, 
	one can view $F_1$ be simply a sum of the size of clauses but it remains to be seen to understand the relevance and potential applications  of higher moments
	such as $F_2$ in the context of CSP. Given the similarity of the core algorithmic frameworks for higher moments, we expect extension of 
	the framework and recipe presented in the paper to derive algorithms for higher moments in the context of CSP. 
	
	\item[Sparse XORs] In the context of model counting, the performance of underlying SAT solvers strongly depends on the size of XORs. The 
	standard construction of $\Hteop$ and $\Hxor$ lead to XORs of size $\Theta(n/2)$ and interesting line of research has focused on the design 
	of sparse XOR-based hash functions~\cite{GHSS07,IMMV15,EGSS14a,AT17,AD16} culminating in showing that one can use hash functions of
	form where $h(x) = Ax+b$ wherein  each entry of m-th row of $A$ is 1 with probability $\mathcal{O}(\frac{\log m}{m})$~\cite{MA20}. Such XORs were
	shown to improve the runtime efficiency. In this context,  a natural direction would be to explore the usage of sparse XORs in the context of $F_0$
	estimation.  
\end{description}

\begin{acks}
We thank the anonymous reviewers of PODS 21 for valuable comments. Bhattacharyya was supported in part by National Research Foundation Singapore under its NRF Fellowship Programme [NRF-NRFFAI1-2019-0002] and an Amazon Research Award. Meel was supported in part by National Research Foundation Singapore under its NRF Fellowship Programme[NRF-NRFFAI1-2019-0004 ] and AI Singapore Programme [AISG-RP-2018-005], and NUS ODPRT Grant [R-252-000-685-13]. Vinod was supported in part by NSF CCF-184908 and NSF HDR:TRIPODS-1934884 awards. Pavan was supported in part by NSF CCF-1849053 and NSF HDR:TRIPODS-1934884 awards. 
\end{acks}
\bibliographystyle{ACM-Reference-Format}
\bibliography{sigproc,streambib}

\end{document}